\documentclass[conference]{IEEEtran}
\IEEEoverridecommandlockouts

\usepackage{cite}
\usepackage{amsmath,amssymb,amsfonts}
\usepackage{algorithmic}
\usepackage{graphicx}
\usepackage{textcomp}
\usepackage{xcolor}
\usepackage{subfig}
\usepackage{amsthm}
\usepackage[normalem]{ulem}
\usepackage{cancel}

\usepackage{mathtools}      
\mathtoolsset{showonlyrefs} 

\usepackage[framemethod=TikZ]{mdframed}
\newmdenv[
  linewidth=0.5pt,
  roundcorner=2pt,
  skipabove=2pt,
  skipbelow=2pt,
  innerleftmargin=6pt,
  innerrightmargin=6pt,
  innertopmargin=6pt,
  innerbottommargin=6pt
]{algotitlebox}

\usepackage{epstopdf}   
\graphicspath{{figs/}}  

\def\BibTeX{{\rm B\kern-.05em{\sc i\kern-.025em b}\kern-.08em
    T\kern-.1667em\lower.7ex\hbox{E}\kern-.125emX}}

\newtheorem{theorem}{Theorem}[section]
\newtheorem{lemma}[theorem]{Lemma}   

\newcommand{\blue}[1]{\textcolor{black}{#1}}

\begin{document}

\title{{\huge Covert Routing with DSSS Signaling Against Cycle Detectors}

\thanks{This material is based on work supported in part by the Army Research
Office / Army Research Laboratory under Contract No. W911NF-23-2-0197.}
}


\author{%
\IEEEauthorblockN{%
Swapnil Saha\IEEEauthorrefmark{1}, 
Rahul Aggarwal\IEEEauthorrefmark{1}, 
Fikadu Dagefu\IEEEauthorrefmark{2},\\
Justin Kong\IEEEauthorrefmark{2}, 
Jihun Choi\IEEEauthorrefmark{2}, 
Brian Kim\IEEEauthorrefmark{2}, 
and Predrag Spasojevi\'c\IEEEauthorrefmark{1}}
\IEEEauthorblockA{\IEEEauthorrefmark{1}WINLAB, Rutgers University, New Brunswick, NJ, USA}
\IEEEauthorblockA{\IEEEauthorrefmark{2}U.S. Army Combat Capabilities Development Command (DEVCOM) Army Research Laboratory, Adelphi, MD, USA}
}


\maketitle

\begin{abstract}
This paper investigates covert multi-hop communication in wireless networks where an adversary employs a cyclostationary (cycle) detector to reveal hidden transmissions. The covert route employs direct-sequence spread spectrum (DSSS) signaling to ensure either maximum end-to-end covertness maximization or minimum latency minimization—under quality-of-service (QoS) and link budget constraints. Optimal bandwidth, transmit power, and spreading gain for each hop jointly satisfy reliability and either rate or covertness requirements. 
We show the equivalence between the covertness- and the {\em detection SNR gain}–based widest-path formulations and, hence,  enabling efficient route computation. Numerical simulations in a realistic $3D$ environment illustrate that $(i)$ end-to-end latency increases exponentially with the covertness requirement, $(ii)$ the end-to-end latency increase is super-linear
with the packet size $M$, and $(iii)$  cycle and energy detectors impose different latency behavior as a function of the message length and the covertness requirement. The proposed framework provides important insights into resource allocation and routing design for covert networks against advanced detection adversaries.

\end{abstract}

\section{Introduction}

\label{sec:intro}
Low probability of detection (LPD) or covert communications \cite{hero2003secure} refers to a communication system that aims to hide its transmissions from a watchful adversary. 
Intuitively, reducing a wireless signal's transmission power and data rate lowers the likelihood of transmissions being detected by unintended observers. 
Multi-hop routing  improves covertness by enabling reduced-power hop-by-hop transmission. 
In addition, Direct Sequence Spread Spectrum (DSSS) modulation is an effective technique for achieving covert communication, as it spreads the message signal across a wider bandwidth, thereby modulating the signal to a level comparable to or exceeding the noise level and providing inherent covertness. Nevertheless, the statistical characteristics of the spreading codes typically exhibit characteristic features/patterns that an adversary can exploit to infer the presence of a transmitted signal \cite{koumpouzi2021improved,koumpouzi2022exploiting,fraisse2023power}.
Beyond security considerations, practical communication systems must also satisfy the Quality of Service (QoS) requirements \cite{aggarwal2025covertness,shi2023user}, such as low latency, high throughput, and reliability.

In~\cite{bash2013limits}, authors take an information-theoretic approach to study communication covertness,
 an approach that does not consider the limitations of a practical detector. 
 Authors in~\cite{yan2019low,kong2022covert,bash2013limits,aggarwal2024covert} study the covertness against an adversary employing an energy detector.
    Adversary that exploits the signal's cyclostationary properties (i.e., spectral correlations) has been studied extensively~\cite{gardner1992signal,aggarwal2025covertness}. We refer to such an adversary as a cycle detector. The authors~\cite{koumpouzi2021improved,koumpouzi2022exploiting,fraisse2023power} identify different strategies that enhance covertness against cycle detectors. 
   In~\cite{aggarwal2025covertness}, we compare the covertness of single link communication systems against  energy and cycle detectors under common communication QoS requirements.
Motivated by the multi-hop routing covertness advantage, several works optimize per-hop resource allocation and routing jointly \cite{xie2022incentive,kong2022covert,sheikholeslami2018multi,kim2024reinforcement} against energy detectors. On the other hand, to the best of our knowledge, no prior research has studied the covertness of multi-hop routing schemes against adversaries employing cycle detectors.
\begin{figure}
    \centering 
    \includegraphics[width=0.9\linewidth]{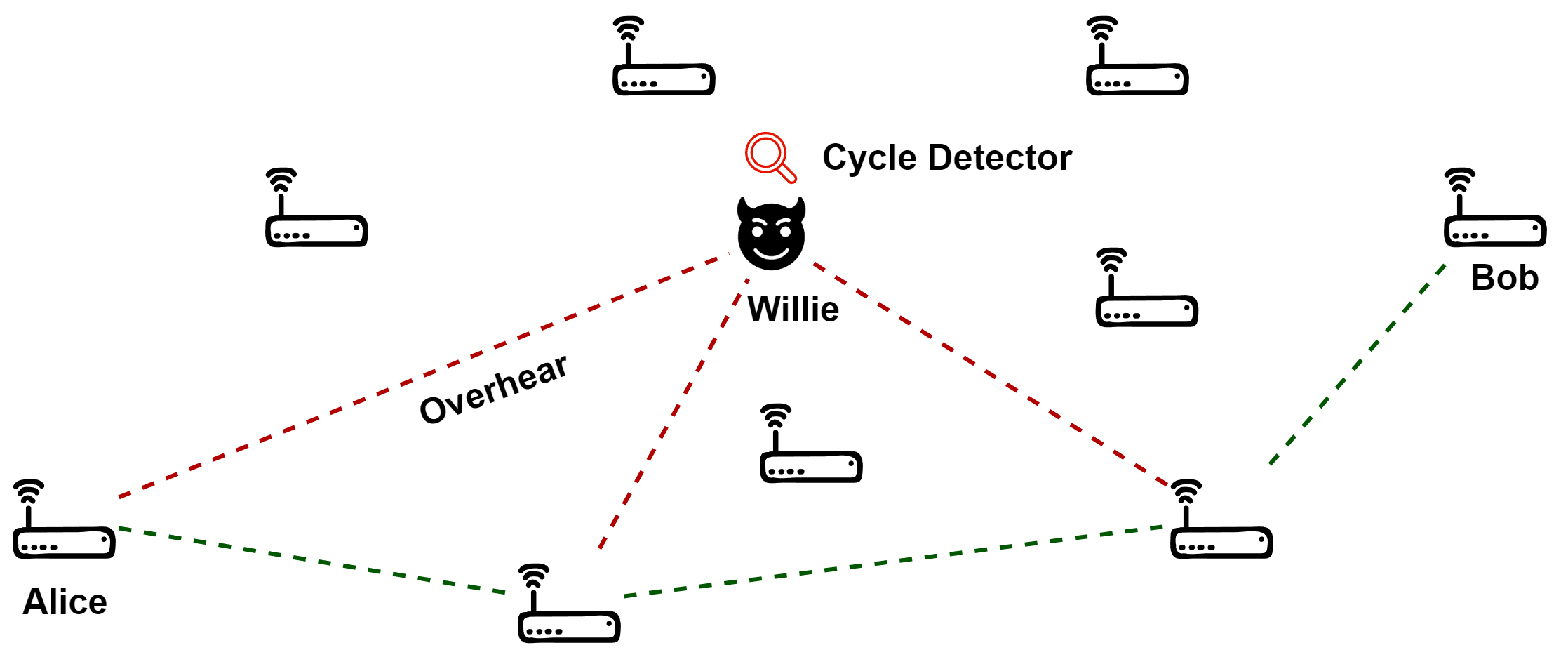}
    \caption{System model of multi-hop network. Each hop uses DSSS signals to communicate with the next hop.}
    \label{fig:system_model}
\end{figure}

In this paper, we examine a wireless covert routing problem comprising Alice, Bob, multiple friendly relaying nodes, and a passive adversary, Willie (Fig.~\ref{fig:system_model}), employing a cycle detector. Alice transmits a confidential message with $M$ bits, modulated by DSSS, to Bob. Alice prefers a lower probability of detection by Willie, with the help of relay nodes whose transmissions satisfy QoS constraints. The message generated by Alice is transmitted hop-by-hop to Bob. Direct communication between Alice and Bob may require high transmitted power due to channel fading and QoS requirements, thereby revealing her communication to Willie. 
Conversely, multi-hop communications may provide more opportunities for the adversary to detect transmissions and/or increase end-to-end message-routing latency. Key contributions are:
\begin{itemize}
    \item We propose two novel routing algorithms for covert multi-hop communication under a cycle detector with QoS guarantees: a covertness-maximization and a latency-minimization approach.

    \item We validated the proposed approach in a realistic ($3D$) environment, revealing that the end-to-end latency  exhibits a sub-linear, linear, and 
    and exponential with an increase in the covertness requirement. 
    Additionally, with an increase in the message size $M$, the route latency exhibits super-linear and exponential growth under relaxed and stringent covertness requirements, respectively.  
    
    \item 
     We also illustrate that, for latency maximization, the adversary may favor the cycle detector at low covert requirements and the energy detector under stringent covert requirements. 
\end{itemize}

\noindent Section~\ref{sec:network_model} gives the legitimate network and signaling model. Section~\ref{sec:covertness_concept} introduces the  per-hop cycle detector.
Section~\ref{sec:perfromace_problem_formulation} formulates the optimization problems. Section~\ref{sec:proposed_sol} describes our proposed method. 
Section~\ref{sec:simulation} gives the numerical results.


\section{Network \& Signal Model}
\label{sec:network_model}

\subsection{Network Model}
As shown in Fig.~\ref{fig:system_model}, Alice can select multiple possible routes to transmit a message to Bob.
We define $\Psi$ as the set of all possible routes. Let us denote $\psi = (v_1, \ldots, v_{N_\psi}) \in \Psi$ as a route from Alice to Bob, where $v_i = (T_{v_i}, R_{v_i})$
is the single-hop communication link between transmitter $T_{v_i}$ and receiver $R_{v_i}$, and $N_\psi$ is the number of hops of the route $\psi$. Here, $T_{v_1}$ and $R_{v_{N_\psi}}$ are Alice and Bob, respectively. Alice generates a message of $M$ bits, which is transmitted sequentially, hop by hop, along the selected route $\psi$ until it reaches Bob. As the analysis of each hop communication is identical, we will drop the hop index $i$ and use $v$ to denote a single-hop communication link. In the network, each of the transmitters $\{ T_{v} \}$ along the path $\psi$ adaptively selects bandwidth $\Omega_v$ and transmitted power $P_v$ based on the dynamic wireless environment, QoS requirements, and covert requirements. There are multiple time slots between any transmitter $T_v$ and $R_v$. Transmitter $T_v$ chooses one of the time slots randomly (discussed in Sec.~\ref{sec:covertness_concept}) and sends the message signal in DSSS format (described in Section~\ref{sec:signal_model}). Willie observes signals from the selected 
hops and decides whether Alice's transmission is present (Willie's detection mechanism is described in Sec.~\ref{sec:willie_detection}). 

\subsection{Single Hop Transmission Model}
\label{sec:signal_model}

In our study, we consider that each node transmitter $T_v$ uses DSSS 
signals to send $M$ bits to the receiver $R_v$ as following. 

\begin{align}\label{equ:dsss_signal}
x_{v}(t) &= \sqrt{P_{v}} \sum_{m=0}^{M-1} b_m 
           \sum_{l_v=0}^{L_v-1} c_{l_v}\, p\!\left(t - mT_{b_v} - lT_{c_v}\right),
\end{align}
where $P_{v}$ denotes the transmit power, $\{b_m\}$ is i.i.d. symbols with equal-probable value $\pm 1$, $T_{b_v}$ is the bit duration, and 
$\{c_{l_v}\}_{l_v=0}^{L_v-1}$ is a bipolar spreading code of length $L_v$. 
The chip duration is $T_{c_v} = T_{b_v}/L_v$, and $p(t)$ is a root raised cosine chip pulse 
with roll-off factor $\beta=1$ and energy $1/L_v$. The signal bandwidth is approximately
$\Omega_v = \frac{1}{T_{c_v}} = \frac{L_v}{T_{b_v}},$
and the uncoded data rate at hop $v$ is $D_v = 1/T_{b_v}$ bps. (For simplicity, in this study, we focus on the uncoded data rate.)
The processing gain is defined as $
\eta_v \triangleq \frac{\Omega_v}{D_v} = L_v,
$ implying that, under perfect despreading, the receiver achieves an SNR gain of $\eta_v$. 
We assume that the intended receiver $R_v$ knows the spreading code 
$\{c_{l_v}\}_{l_v=0}^{L_v-1}$, whereas Willie does not. Accordingly, the SNRs are
\begin{align}\label{equ:SNR}
    SNR_v &= \frac{P_v |h_v^{T_v,R_v}|^2}{N_0 \Omega_v}\, \eta_v, \,\,\,
    SNR_v^{W}   = \frac{P_v |h_v^{T_v,W}|^2}{N_0 \Omega_v}.
\end{align}
Here, $N_0/2$ denotes the power spectral density (PSD) of the additive white Gaussian noise.  $h_v^{T_v,R_v}$ and $h_v^{T_v,W}$ represent the channel gains between the transmitter $T_v$ and its intended receiver $R_v$, and between $T_v$ and Willie, respectively. Finally, the link reliability between $T_v$ and $R_v$ is quantified by the bit error rate (BER) $BER_v = Q\!\left(\sqrt{2 \times SNR_v}\right),$ where $Q(\cdot)$ is the complementary Gaussian distribution.

\section{Covertness Against Cycle Detector}\label{sec:covertness_concept}

\subsection{Background: Cycle Detector}
If a signal has an inherent periodic structure, the Degree of Cyclostationarity (DCS) \cite{gardner1992signal} is an excellent feature for detecting the presence of the signal. The DSSS signal, with the same spreading code repeating, has a higher DCS value than the white noise signal. The DCS metric is calculated as follows. For a periodic autocorrelation function $R_y(t,\tau)=\mathbb{E}{[y(t)y^*(t-\tau)]}$ with period $T$, the Cyclic Autocorrelation Function (CAF) is $R_y^\alpha \triangleq \frac{1}{T} \int_{-T/2}^{T/2} R_y(t,\alpha)e^{-i2\pi\alpha t}dt$ where $\alpha=\frac{k}{T}$, for $k=1,2,3..$ are the cycle frequencies. The Spectral Correlation Function (SCF) of $y(t)$, $S_y^{\alpha}(f)$, is the Fourier Transform of $R_y^{\alpha}(\tau)$. Then the DCS is defined as 

\begin{equation}\label{equ:DCS_cal}
    DCS^y \triangleq \sum_{\alpha \neq 0} \frac{\int_{-\infty}^{+\infty} |S_y^{\alpha}(f)|^2df}{\int_{-\infty}^{+\infty} |S_y^{\alpha=0}(f)|^2df}.
\end{equation}
For a finite segment of signal $y(t)$, $t \in (0,T_0)]$, the SCF is estimated as follows $\hat{S}_{Y}^\alpha(f) = Y_{T_0}\left(f - \frac{\alpha}{2}\right) Y_{T_0}^*\left(f + \frac{\alpha}{2}\right)$,
where $Y_{T_0}(t, f) = \int_{t - T_0/2}^{t + T_0/2} y(u) e^{-i2\pi f u} \, du$
is time-varying finite-time complex spectrum~\cite{gardner1989statistical}.
Plugging the estimate SCF, $\hat{S}_{Y}^\alpha(f)$, in Eq. \eqref{equ:DCS_cal}, we calculate the estimated DCS, $\widehat{DCS}^y.$ 

\subsection{Single Hop Cycle Detection at Willie}\label{sec:willie_detection}
For single link transmission,  at hop $v,$  transmitter $T_v$  chooses to transmit or not with probability $1/2$. We assume that Willie has the knowledge of when the transmission slots start and ends. This assumption corresponds to the worst-case scenario from the perspective of covertness.  Willie selects one of the following two hypotheses:
\begin{align*}
H^0_v &: y^W_v(t) = n^W_v(t), \; t \in (0, T_0), \\
H^1_v &: y^W_v(t) = h^{T_v,W}_{v} x_v(t) + n_v^w(t), \; t \in (0, T_0),
\end{align*}
where $T_0=M T_{b_v}$ is the slot duration; $h^{T_v,W}_{v}$ is the channel gain form transmitter $T_v$ to Willie;
$H_v^0$ is the no transmission null hypothesis; and $H_v^1$ denotes the alternate hypothesis when the transmission between $T_{v}$ and $R_{v}$ is present. To detect the presence of the hop transmission, Willie estimates the strength of the cyclic feature by calculating the DCS value. The white noise does not have cyclic features, whereas the DSSS signal at Eq.~\eqref{equ:dsss_signal} does. The strength of cyclic properties can be quantified via Degree of Cyclostationarity (DCS)\cite{gardner1988cyclic,aggarwal2025covertness}.
Willie compares the statistics $\widehat{DCS}^W_v \underset{{\mathcal{D}}^0_{v}}{\overset{{\mathcal{D}}^1_{v}}{\gtrless}} \delta_{v},$
where $\widehat{DCS}^W_v$ is the estimated DCS of  $y_v^W(t)$ and $\delta_v$ is the decision threshold.
$\mathcal{D}^0_{v}$ and $\mathcal{D}^1_{v}$ denote decisions in favor of $H_v^0$ and $H_v^1$. 
The detection error probability (DEP) at Willie for link $v$ is $ DEP_v  = (P^{MD}_v + P^{FA}_v)$,
where $P^{\text{MD}}_{v} \triangleq \mathbb{P}\left(\mathcal{D}^0_{v} \mid H^1_v \right)$  and $P^{\text{FA}}_{v} \triangleq P\left(\mathcal{D}^1_{v} \mid H^0_v \right)$ are the probabilities of miss detection and false alarm, respectively. Further, we omit the prior probability factor $.5$  without loss of generality. 


\label{sec:min_DEP_Willie}
Willie's goal is to minimize each hop's detection error probability, $DEP_v$, by selecting the optimum DCS threshold $\delta_V^{opt}$. It selects the threshold ~\cite{aggarwal2025covertness} : 
\begin{equation}\label{equ:delta_opt}
    \delta_v^{opt} = \arg\min_{\delta_v} \; DEP_v (SNR_v^W).
\end{equation}
where $SNR_v^W$ is the SNR at Willie when the signal is present.
Here, we assume that Willie has the perfect knowledge of all cycles in $x_v(t)$,  and the $SNR_v^W$ (when present and when not) so that it can compute the $DEP_v$ (worst-case scenario from a covertness perspective). Following \cite{aggarwal2025covertness}, the $\delta_v^{opt}$ is obtained via exhaustive search due to the lack of an efficient analytical approach. 
In \cite{aggarwal2025covertness}, we argue that $DEP_v \propto \frac{1}{SNR_v^W}$, which is due to the fact that increased SNR at Willie both ensures increased {\em distance} between the two hypotheses measured via the DCS as well as an increased accuracy in estimating DCS. We also show that DEP decreases with observation time, i.e., $DEP_v \propto \frac{1}{M}$. In the above, $\propto$ indicates a monotonic non-decreasing relationship between parameters and statistics, which we exploit in this paper.

\section{Performance Metrics \& Problem Formulation}\label{sec:perfromace_problem_formulation}

We aim to study two routing optimization problems in which Alice sends $M$ bits to Bob. The first one is 
how covertly the $M$ bits reach Bob from Alice; 
end-to-end covert maximization. The second one is how quickly $M$ bits can be transferred to Bob, with a focus on minimizing end-to-end latency. For both routing optimization problems, we impose constraints to ensure QoS requirements. We will first quantify the performance metrics for each problem and then provide an explicit description of our problem formulation. Here, for quantifying end-to-end covertness, we assume that Willie has complete knowledge of the optimized route from Alice to Bob, $\psi = (v_1, v_2, \ldots, v_{N_\psi}) \in \Psi$. This represents the worst-case condition for covert communication and the most favorable case for Willie.
\subsection{Performance Metrics}
\textit{End-to-end Covertness $DEP(\psi)$}: We define the end-to-end covertness of a route $\psi$ as follows. Based on the definition of covertness in Sec.~\ref{sec:willie_detection}, each hop's transmission is independent from others. If any hop $v \in \psi$ is \emph{weak} in covertness and cannot {\em hide} communication, the end-to-end covertness collapses, and Willie can detect Alice's transmissions. Thus, end-to-end covertness of the route $\psi$ is constrained by the weakest link, i.e., $DEP(\psi) = \min_{v \in \psi} DEP_v$.\\
\textit{End-to-end Latency $\lambda(\psi)$}: For a single hop $v$, we define the  communication latency $\lambda_v$ as the expected  time to send $M$ bits to the next hop, i.e., $\lambda_v=\frac{M}{D_v}$ where $D_v$ is the link data rate (uncoded). The end-to-end/total latency for the route $\psi$ is the sum of all the link latencies, $\lambda(\psi)=\sum_{v \in \psi} \lambda_v$. We note that this formulation provides a simplified estimate of overall latency, as it ignores queuing delays and protocol overheads of hop-to-hop data transmission. 

\subsection{End-to-End Covertness Maximization}
\label{sec:max_covert_route}

The end-to-end covertness can be maximized by maximizing the weakest DEP in the route $\psi:$
\begin{equation}\label{equ:op_ete_covert}
\begin{aligned}
\max_{P_{v}, \eta_{v}, \psi \in \Psi} \;    
& \min_{v \in  \psi} \;  DEP_v  , \\
\text{s.t.} \quad & D \geq D^{\mathrm{reqd}} , \\
[BER_v,\Omega_v,P_v] \ \preceq & \; [BER^{reqd},\Omega^{max},P^{max}]
\end{aligned}
\end{equation}
Here, $D^{reqd}$ is the (average) required link data rate;  $BER^{reqd}$, the bit error rate,  ensures link reliability; 
$\Omega^{max}$ and $P^{max}$ are the bandwidth and transmitted power budgets of a single link. 
\subsection{End-to-End Latency Minimization}
 We minimize the overall latency 
\begin{equation}\label{equ:op_ete_latency}
\begin{aligned}
\min_{P_{v}, \eta_{v}, \psi \in \Psi} \; & \sum_{v \in \psi} \lambda_v , \\
\text{s.t.} \quad & DEP_v \geq DEP^{\mathrm{reqd}} , \\
[BER_v,\Omega_v,P_v] \ \preceq & \; [BER^{reqd},\Omega^{max},P^{max}]
\end{aligned}
\end{equation}
under the constraint that each link satisfies the common required covertness $DEP^{reqd}$ as well as the above-mentioned QoS and link budget constraints.


\section{Proposed Solution}
\label{sec:proposed_sol}
\blue{Based on the system model from Sec.~\ref {sec:network_model}, metrics (and constraints) for each hop in a route  are independent. Thus, end-to-end optimality can be achieved by first achieving link optimality, followed by optimum route selection. In the next section, we will first describe our approach for single-link optimization, followed by optimal route selection.}
\subsection{Single Hop Covert Maximization}\label{sec:single_hop_covert}

\begin{equation}\label{equ:link_covert}
\begin{aligned}
\max_{P_{v}, \eta_{v}} \; &  DEP_v , \\
\text{s.t.} \quad & D_v \geq D^{ \mathrm{reqd}} , \\
[BER_v,\Omega_v,P_v] \ \preceq & \; [BER^{reqd},\Omega^{max},P^{max}]
\end{aligned}
\end{equation}
Solution of~\eqref{equ:link_covert} is detailed in our previous work~\cite{aggarwal2025covertness}. Here, we summarize the results in the context of this study. 
Eq.~\eqref{equ:link_covert} is maximized by the optimum spreading gain  $\eta^{opt} =\Omega^{max}/D^{reqd}$  and the link transmission power 
\begin{equation}
P_v^{opt} = \frac{SNR^{reqd}}{|h_v^{T_v,R_v}|^2}\cdot\frac{N_0\Omega_{max}}{\eta^{opt}}.
\end{equation}
both  achieved by selecting maximum bandwidth $\Omega_v^{opt}=\Omega^{max}$ and  the minium required link rate $D_v^{opt}=D^{reqd}$.  Here, SNR of Bob's received signal, $SNR^{reqd},$ ensures that the  reliability constraint, $BER^{reqd},$
common to all links, is satisfied with equality.
The maximum $DEP_v^{max}$ is achieved by the minimum detection SNR at Willie, $SNR_v^{W,min}$ which maximizes the detection SNR gain $\theta_v \triangleq \frac{SNR_v}{SNR_v^W}$ and achieves
\begin{equation}
    \theta_v^{max} = \frac{SNR^{reqd}}{SNR_v^{W,min}} = \frac{|h_v^{T_v,R_v}|^2}{|h_v^{T_v,W}|^2} \times \eta^{opt}. 
\end{equation}
\subsection{Single Hop Latency Minimization}
Single-hop latency minimization requires
\begin{equation}\label{equ:link_latency}
\begin{aligned}
& \min_{P_{v}, \eta_{v}} \; \lambda_v , \\
\text{s.t.} \quad & DEP_v \geq DEP^{\mathrm{reqd}} , \\
[BER_v,\Omega_v,P_v] \ \preceq & \; [BER^{reqd},\Omega^{max},P^{max}].
\end{aligned}
\end{equation}
The link covertness  and reliability requirements impose  the opposing  SNR constraints  $SNR^{W}_v \leq SNR_v^{W,max}$ and 
$SNR_v \geq SNR_v^{reqd}$ at Willie and Bob, respectively. Together they require that the link transmitted power $P_v$ satisfies 
\begin{equation}
\label{equ:P_v_lat_min}
      \frac{SNR_v^{req}} {|h_v^{Tv,Rv}|^2}\cdot D_v   \leq \frac{P_v}{N_0 \Omega}  \leq      \frac{SNR_v^{W,max}} {|h_v^{Tv,W}|^2}.
\end{equation}
And, hence,
\begin{equation}
   \lambda_v 
  \ge M\,\frac{SNR_v^{\mathrm{req}}}{|h_v^{T_v,R_v}|^{2} P_v}\,N_0 
  \ge M\,\frac{SNR_v^{\mathrm{req}}}{SNR_v^{W,\max}}
       \frac{|h_v^{T_v,W}|^{2}}{|h_v^{T_v,R_v}|^{2}}
       \frac{1}{\Omega}.
       \label{equ:L_v_lat_max} 
\end{equation}


Hence, the maximum  power and minimum link  latency,  
\begin{align}
P_v^{\text{max,reqd}}
  &= \frac{\mathrm{SNR}_v^{W,\max}}{\lvert h_v^{T_v,W}\rvert^2}
     \, N_0 \, \Omega^{\max}, \label{eq:Pmax} \\[6pt]
\lambda_v^{\min}
  &= M \frac{\mathrm{SNR}_v^{\mathrm{reqd}}}{\mathrm{SNR}_v^{W,\max}}
     \frac{\lvert h_v^{T_v,W}\rvert^2}{\lvert h_v^{T_v,R_v}\rvert^2}
     \frac{1}{\Omega^{\max}}. \label{eq:lambda_min}
\end{align}
ensure all constraints are satisfied.
The following lemma directly follows the above discussion and  inequality~\eqref{equ:L_v_lat_max}. 
\footnote{For simplicity without loss of generality we assume that $P_v^{max,reqd}\leq P^{max}$ and that the spreading gain $\eta_v$ is not necessarily an integer.}
\begin{lemma}\label{thm:omega_processing_gain}
Let M be the packet size.
    The solution of~\eqref {equ:link_latency} is the latency $\lambda_v^{min},$ achievable with  bandwidth $\Omega^{max},$ transmit power $P_v^{max,reqd},$   and  spreading gain $\eta_v^{opt}=\Omega^{max} \lambda_v^{min}/M.$  
\end{lemma}
\noindent{\bf Remark:}  The optimal transmit powers are determined by the minimum detection SNR  required for link reliability at Bob's, $SNR_v^{reqd},$ and by the maximum SNR still ensuring link covertness at Willie's,  $SNR_v^{W,max},$ for the  covertness and latency 
optimization problems, respectively.  
The optimal spreading gains are determined by the required link rate,  $D^{reqd},$ and  the minimum {\em reliable and covert} latency, $\lambda_v^{min},$ for the  covertness and latency 
optimization problems, respectively.  
\label{sec:single_hop_latency}
\subsection{Covert Routes: Covertness  vs Latency Optimization}
The  {\em route that maximizes covertness} is the solution of:
\begin{equation}\label{equ:max_min_dep}
\begin{aligned}
\max_{\psi \in \Psi} \;    
 \min_{v \in  \psi} \; & DEP_v^{max}.
\end{aligned}
\end{equation}
We solve this problem using a Dijkstra-based widest path algorithm~\cite{medhi2017network}, which maximizes the minimum edge weight along the route by replacing the path-length metric with a bottleneck distance metric. 
Consequently, the edge with the lowest weight 
has to be found first. As mentioned in  Sec.~\ref{sec:min_DEP_Willie}, the $DEP_v$ computation is based on the empirical distributions under the null and alternative hypotheses obtained via Monte Carlo simulation.
which is computationally intensive.
Instead, here, we propose 
to solve  the equivalent problem 
\begin{equation}\label{equ:max_min_theta}
\begin{aligned}
\max_{\psi \in \Psi} \;    
 \min_{v \in  \psi} \; & \theta_v^{max},
\end{aligned}
\end{equation}
based on the detection SNR gain parameter $\theta_v.$ 
\begin{theorem}\label{thm:equivalent_problem}
   The routes  \eqref{equ:max_min_dep} and \eqref{equ:max_min_theta} are equivalent.
\end{theorem}
\begin{proof}
The proof relies on $\theta_v \propto DEP_v$  in problem \eqref{equ:link_covert} (see~\cite{aggarwal2025covertness} and~\ref{sec:single_hop_covert}).
The reliability constraint in~\eqref{equ:link_covert} requires that the optimal detection
SNR at Bob's is identical
($SNR_v^ \equiv SNR^{reqd}$) 
for all links and, additionally,   $DEP_v \propto \frac{1}{SNR_v^W}$ (see Sec.~\ref{sec:min_DEP_Willie}).
Consequently,
$DEP_v$ is a non-decreasing function of $\theta_v.$
Now, let $(v_1,\ldots,v_N)$ index the hops ordered according to
\[
\mathrm{DEP}_{v_1}\;\le\;\mathrm{DEP}_{v_2}\;\le\;\cdots\;\le\;\mathrm{DEP}_{v_N}.
\]
Then,  the
corresponding gains are ordered as
\[
\theta_{v_1}\;\le\;\theta_{v_2}\;\le\;\cdots\;\le\;\theta_{v_N}.
\]
Hence, the equivalence of the two problems follows.
\end{proof}

The {\em covert route that minimizes  latency} is the solution of 
\begin{equation}\label{equ:min_lat_route}
    \min_{\psi \in \Psi} \sum_{v \in \psi} \lambda_v^{min},
\end{equation}
and is solved using Dijkstra’s shortest path algorithm.
\section{Numerical Results}
\label{sec:simulation}
In this section, we present numerical simulation results to validate the system's performance under our proposed solution\footnote{The simulation code is available at https://github.com/swapnil-saha/Covert-Routing-with-DSSS-Signaling-Against-Cycle-Detectors.git}. We perform our joint resource allocation and routing within an emulated $ 3D$-propagation environment. Unless otherwise specified, all simulation results are obtained using the following parameter settings: maximum available bandwidth is $\Omega_{\max} = 10~MHz,$ 
the noise spectral density is $N_0 = -113~\mathrm{dBm/Hz}$, and message size $M=100~Mb$. We choose $BER^{reqd} =Q(\sqrt{2\times10}) \xrightarrow{} SNR^{reqd}=10~\mathrm{dB}$ as constraint   for reliable communication.\\
\begin{figure}[t]
    \centering
    \includegraphics[width=0.65\linewidth]{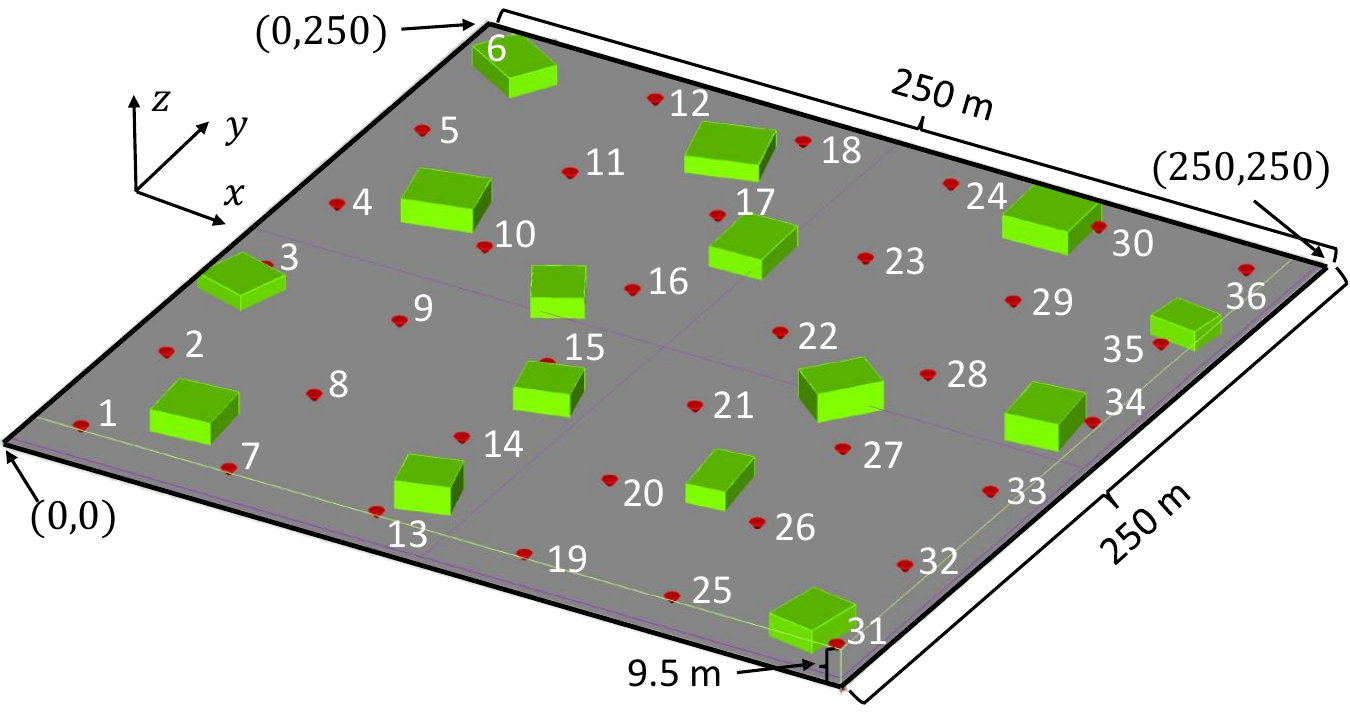}
    \caption{Network topology.}
    \label{fig:sim_model}
    
    \includegraphics[width=0.75\linewidth]{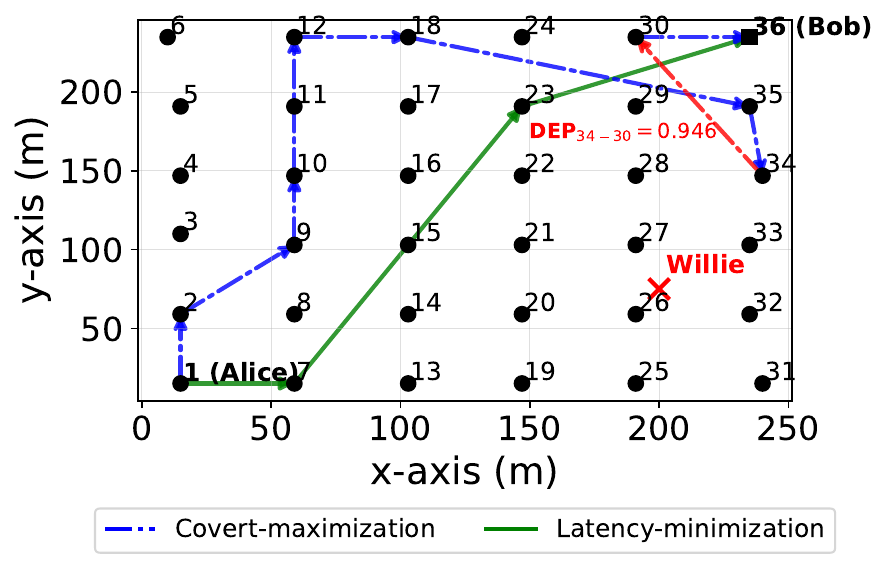}
    \caption{Optimized routing.}
    \label{fig:routing}
\end{figure}
As a network topology (Fig.~\ref{fig:sim_model}), we use an emulated $3D$  environment $(250 \times 250 \times 9.5) \; m^3$ with $36$ nodes (shown in red cones). The green cuboids indicate concrete buildings. We use a ray-tracing approach \cite{EMAGTech} to obtain channel data, assuming a center frequency of $900~MHz$ for the transmitted signal. Readers are referred to our previous work \cite{kong2022covert} for details on the mechanism for obtaining the channel gains for this network topology. In our simulation, we assume that Alice is at node $1$, Bob is at node $36$, and Willie is at the position between node $27$ and node $26$ (Fig.~\ref{fig:routing}).


\subsection{Optimal Covert Routes}
\label{sec:sim_opt_route}
Fig.~\ref{fig:routing} illustrates the optimized routing paths obtained from the solutions of \eqref{equ:op_ete_covert} and \eqref{equ:op_ete_latency}. For the covert maximization problem~\eqref{equ:op_ete_covert}, each link’s required throughput is set to $D_v^{\mathrm{reqd}}=2.5~\mathrm{Mbps}$, while for the latency minimization problem ~\eqref{equ:op_ete_latency}, the required link covertness is fixed at $DEP_v^{\mathrm{reqd}}=0.85$. For the covertness-maximizing route, the link covertness values 
are $[0.980, 0.979, 0.971, 0.981, 0.980, 0.960, 0.950,\textbf{0.946},0.949]$, yielding an end-to-end covertness of $DEP(\psi)=0.946$. For the latency-minimizing route, the total time to transmit a message of size $M=100~\mathrm{Mb}$ is $30~\mathrm{s}$. Two main observations can be drawn: (i) Transmitters located closer to Willie achieve lower link covertness than those that are far away: nodes $30$, $35$, and particularly $34$ (the bottleneck hop shown by the red dotted link). (ii) The latency-minimizing route employs fewer hops and lies closer to Willie. This is intuitive: in the covertness-maximization formulation, the optimization prioritizes per-link covertness, whereas in the latency-minimization formulation, the relaxed covertness constraint permits higher transmit power, enabling longer transmission distances between hops and thus reducing end-to-end delay. Moreover, optimizing routes of both problems ~\eqref{equ:max_min_dep} and \eqref{equ:max_min_theta} provides the same optimum routing (not shown in the figure), confirming the claim of Theorem~\ref{thm:equivalent_problem}.

\subsection{Impact of Link Covertness Requirement on Route Latency}\label{sec:impact_of_constraint}
\begin{figure}[t]
    \centering
        \subfloat[]{%
        \includegraphics[width=0.8\linewidth]{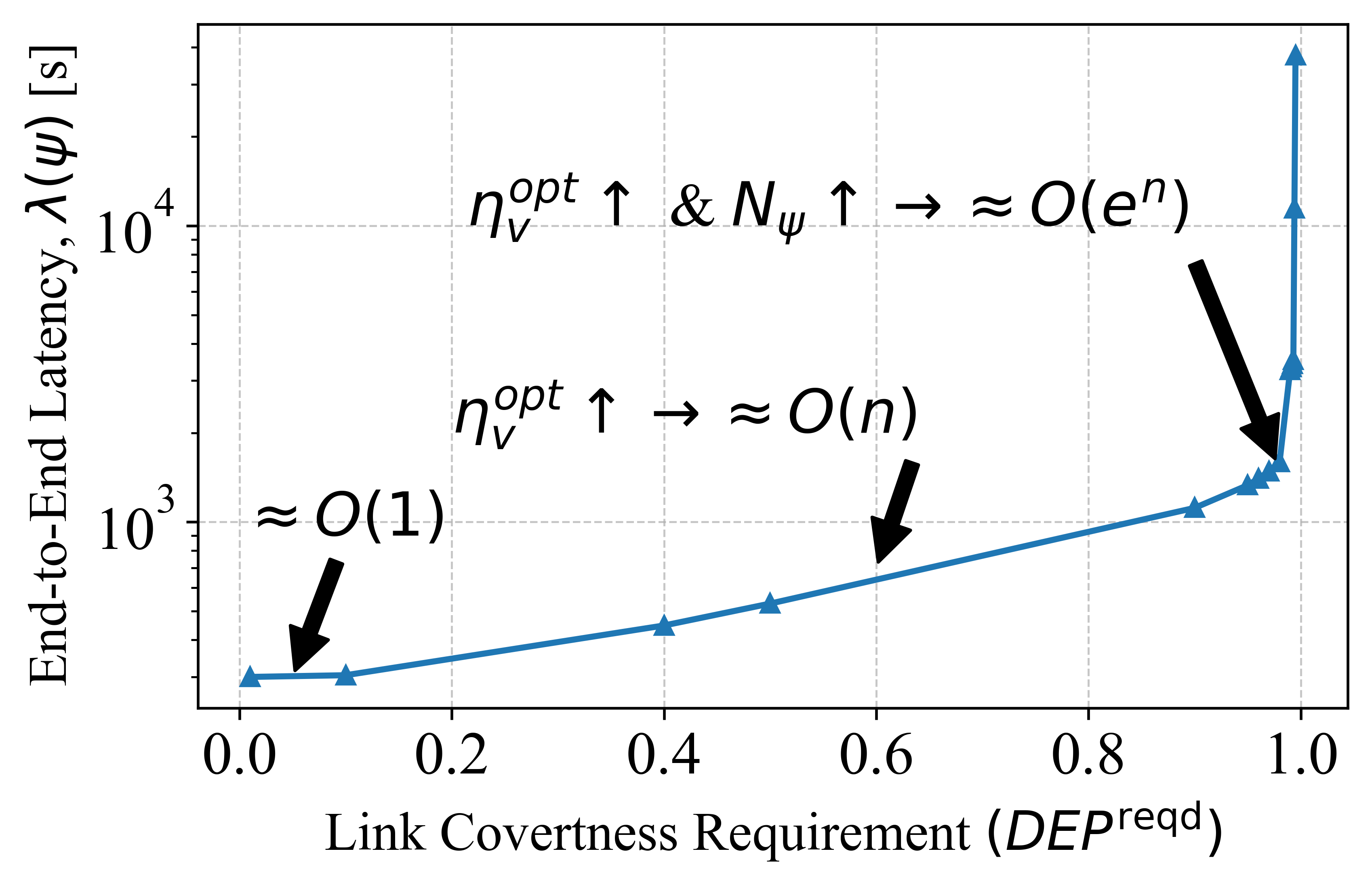}%
        \label{fig:lat_vs_dep}}
    \vspace{0.3em}
    \subfloat[]{%
        \includegraphics[width=0.5\linewidth]{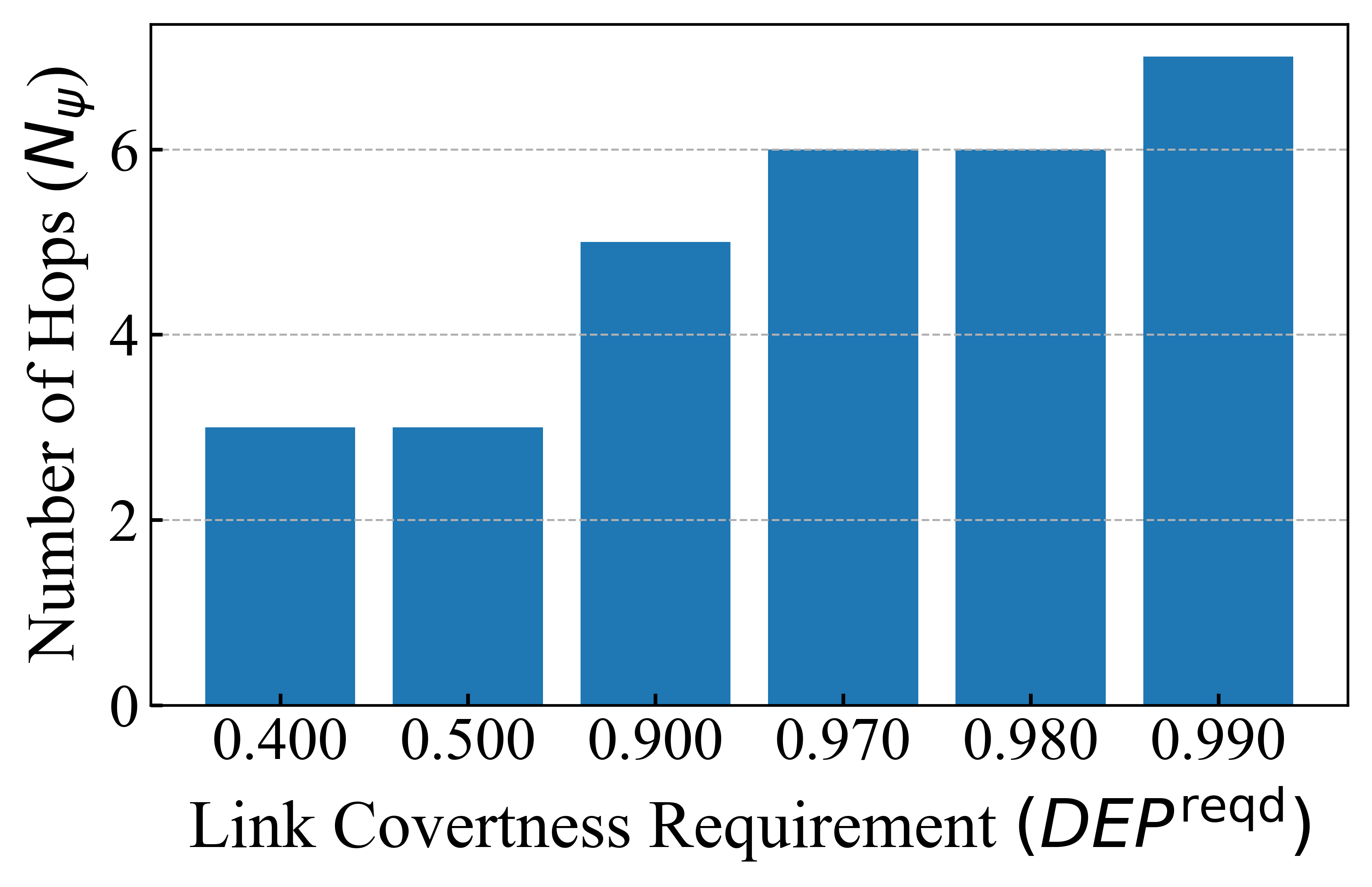}%
        \label{fig:hop_vs_dep}}
    \hfill
    \subfloat[]{%
        \includegraphics[width=0.5\linewidth]{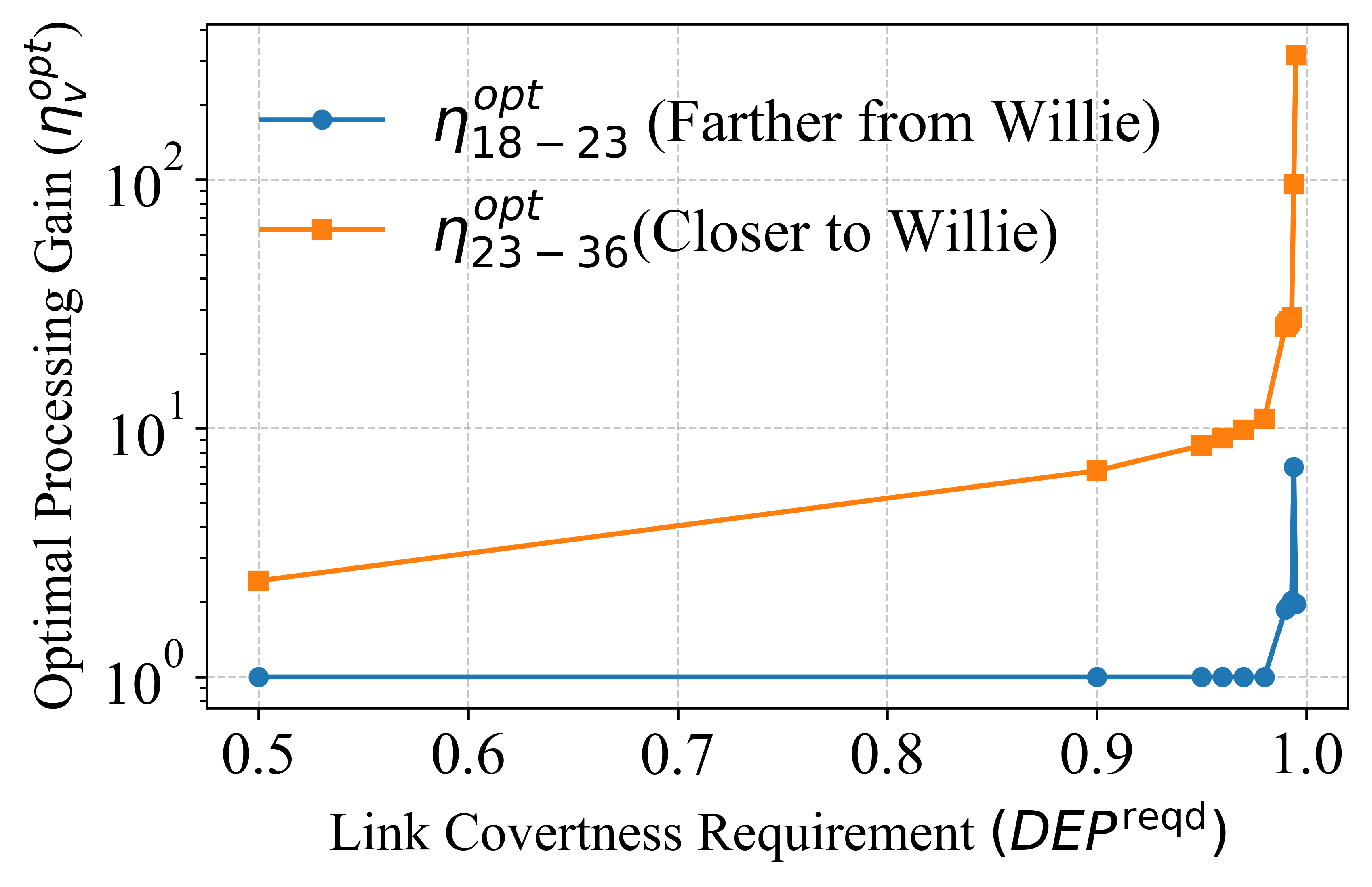}%
        \label{fig:eta_vs_dep}}
    \caption{End-to-end latency versus link covertness requirement, showing a transition from constant $O(1)$ to linear $O(n)$ and near-exponential growth as stricter covert constraints force increases in both the hop number and processing gain.}
    \label{fig:link_covertness_on_latency}
\end{figure}
Fig.~\ref{fig:link_covertness_on_latency} 
illustrates, for our network topology and requirements, the following results: for $DEP^{reqd} < 0.1$, latency remains nearly constant; for $0.1 \leq DEP < 0.97$, latency increases approximately $O(n)$; and for $DEP \geq 0.97$, it exhibits exponential growth. The increase of route latency is triggered by two factors: increasing the required number of hops $N_{\psi}$ and increasing the optimal processing gain $\eta_v^{opt}$ of hops close to Willie. The intuitive explanations of this behavior are as follows. A low covertness requirement allows higher transmission power at each node, thereby enabling near-direct communication between Alice and Bob with fewer hops (See Fig.~\ref{fig:hop_vs_dep}).  
At low covertness requirement, the optimal processing gain remains constant $\eta_v^{opt} \approx 1$, where hop $v$ is closer to Alice  (See Fig.~\ref{fig:eta_vs_dep}). Since $D_v \propto \frac{1}{\eta_v},$
latency also remains constant. 
The increase in the link covertness requirement has two significant effects. The optimized route requires more hops (with a closer proximity to Willie) {\em and} the required/optimum per hop processing gain $\eta_v^{opt}$ increases. For a moderate covert requirement, the latency exhibits approximately linear growth as the optimal transmission parameter $\eta_v^{\mathrm{opt}}$ increases (see Fig.~\ref{fig:eta_vs_dep}). Consequently, once a certain threshold $DEP^{reqd}$ is crossed, the minimum achievable latency grows exponentially.

\subsection{DCS  vs Energy Metric}
\begin{figure}[t]
  \centering
\subfloat[$DEP^{reqd}=10^{-2}$]{%
\includegraphics[width=0.70\columnwidth]{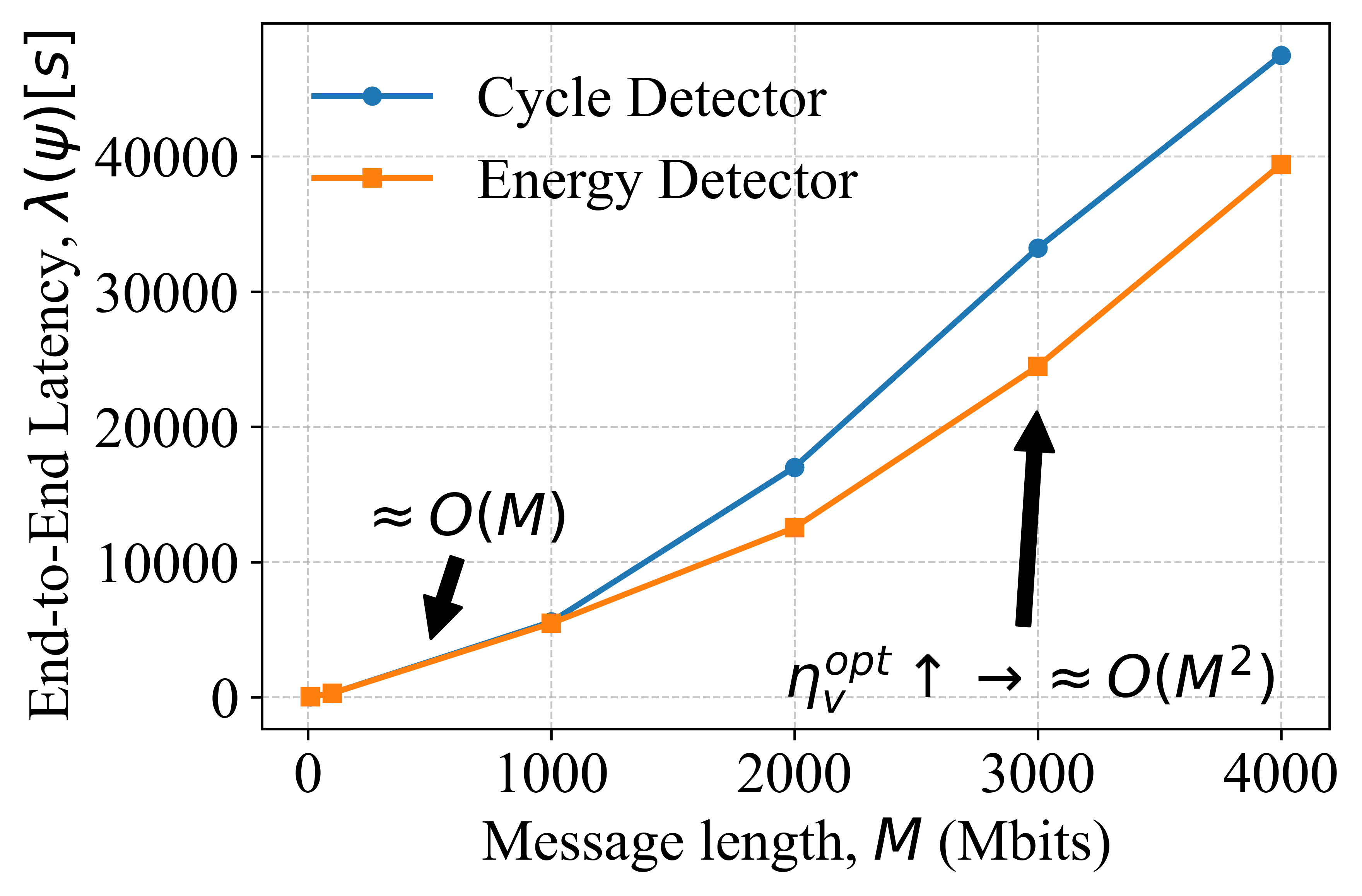}%
\label{fig:low_DEP}}
  \hfill  \subfloat[$DEP^{reqd}=0.8$]{%
\includegraphics[width=0.70\columnwidth]{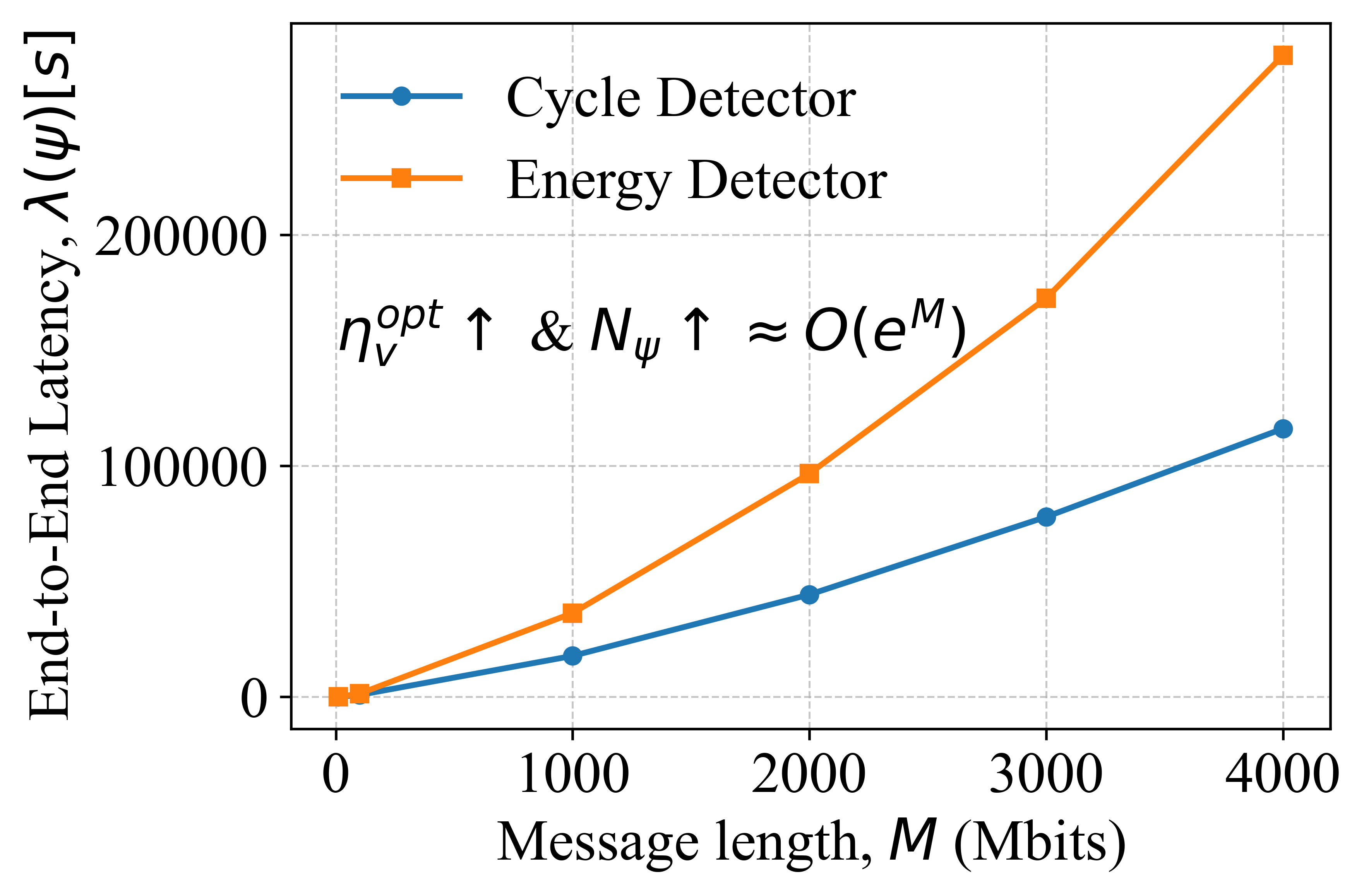}%
\label{fig:high_DEP}}
  \caption{End-to-end latency of cycle and energy detectors versus message length under relaxed and stringent covert requirements. For low $DEP$, the cycle detector imposes routes with higher latency than the energy detector and vice versa.}
\label{fig:lat_through_M}
\end{figure}
Here, we quantitatively compare the energy and cycle detectors in the context of the latency minimization problem (Fig.~\ref{fig:lat_through_M}). The energy detector implementation follows our prior work~\cite{kong2022covert}. We vary two key system parameters: the link covertness requirement $DEP^{reqd}$ and the message length $M$. The results reveal two main observations:
$(i)$ For latency minimization, the system performance decreases with the message length $M$;  
$(ii)$ To maximize the best route's latency, Willie would select the cycle detector when Bob's covertness requirement is low and the energy detector when it is high. 
The explanation closely follows that of our prior study~\cite{aggarwal2025covertness}. Since $DEP_v \propto \frac{1}{M}$ (Sec.~\ref{sec:willie_detection}), the transmit power $P_v$ needs to decrease to guarantee the required link covertness as $M$ increases. 
For the lower link covertness requirement (Fig. \ref {fig:low_DEP}), the optimal route follows an approximately direct path from Alice to Bob.
For $M \leq 100~Mbs$, optimal processing gain is constant and $\eta_v^{opt}\approx 1$ 
and, thus, latency increase is $O(M)$. 
For $M > 1000~Mbs$, 
the processing gain increases in a linear fashion with $M$, and, consequently, the latency increase is roughly $O(M^2).$ On the other hand, for a higher covertness requirement, 
due to both an increase in the required number of hops (Fig.~\ref{fig:hop_vs_dep}) and an increase in the $\eta_v^{opt}$ 
the optimal route experiences exponential growth vs 
$M.$
Different behavior of the energy and the cycle detector as a function of these two parameters is due to the difference in order of their respective statistics: energy detection relies on second-order statistics, whereas DCS is a fourth-order statistic, Eq.~\eqref{equ:DCS_cal}, making it more noise-sensitive (also see Fig.~$1$ in \cite{aggarwal2025covertness}). Hence, tightening the covertness requirement $DEP^{reqd}$ reduces the required transmit power and degrades both signal estimation and detection, disproportionately affecting the cycle detector. For our network topology and QoS requirements, we found that the transition at which the cycle degrades relative to the energy detector occurs at approximately $DEP^{reqd} \approx 10^{-1}$.

\bibliographystyle{IEEEtran}  
\bibliography{refs}           

\end{document}